\documentclass[preprint,10pt]{elsarticle}

\usepackage[dvipsnames]{xcolor}
\usepackage{colortbl}
\usepackage[title]{appendix}
\usepackage{graphicx}
\usepackage{amsfonts}
\usepackage{url}
\usepackage{amscd}
\usepackage{amssymb}

\usepackage[algosection,lined,boxed,commentsnumbered,ruled,vlined]{algorithm2e}
\usepackage[hidelinks]{hyperref}
\usepackage[dvipsnames]{xcolor}

\usepackage{indentfirst}
\usepackage[greek,english]{babel}
\usepackage[utf8]{inputenc}
\usepackage{color,amsmath,amssymb,graphicx}
\usepackage{tikz-cd}
\usepackage{amsmath,amssymb}
\usepackage{enumerate}
\usepackage{dsfont}
\usepackage{appendix}
\usepackage{amsthm}
\usepackage{float}
\usepackage[algosection,lined,boxed,commentsnumbered,ruled,vlined]{algorithm2e}
\usepackage{mathtools}

\newtheorem{theorem}{Theorem}[section]
\newtheorem{lemma}[theorem]{Lemma}	
\newtheorem{corollary}[theorem]{Corollary}
\newtheorem{definition}[theorem]{Definition}
\newtheorem{proposition}[theorem]{Proposition}
\begin{document}

\begin{frontmatter}

\title{Cryptographic Primitives based on Compact Knapsack Problem}
\author[1]{George S. Rizos} 
\ead{georgiosrs@csd.auth.gr}
\author[1]{Konstantinos A. Draziotis}
\ead{drazioti@csd.auth.gr}

\address[1]{Department of Informatics\\
Aristotle University of Thessaloniki \\
54124 Thessaloniki, Greece}

	\begin{abstract}
	In the present paper, we extend previous results of an id scheme based on compact knapsack problem defined by one equation. We present a sound three-move id scheme based on compact knapsack problem defined by an integer matrix. We study this problem by providing attacks based on lattices. Furthermore, we provide the corresponding  digital signature obtained by Fiat-Shamir transform and we prove that is secure under ROM. These primitives are post quantum resistant.
	\end{abstract}

\begin{keyword}
Public Key Cryptography\sep Lattices\sep Closest Vector Problem\sep Babai's Nearest Plane Algorithm\sep Zero Knowledge\sep Sigma-id-schemes \sep Fiat-Shamir transform \sep Compact knapsack problem\sep Digital signatures\sep Linear Systems.
%% keywords here, in the form: keyword \sep keyword

%% PACS codes here, in the form: \PACS code \sep code

%% MSC codes here, in the form: \MSC code \sep code
%% or \MSC[2008] code \sep code (2000 is the default)
\MSC[2020] 94A60 %\sep 11H06\sep 11T71\sep 68R01.
\end{keyword}
\end{frontmatter}
{\let\thefootnote\relax\footnotetext{All the authors contributed equally to this research.}}
\section{Introduction}
Before we proceed presenting our results we define some terminology and give some necessary definitions. 
With Compact knapsack problem, we mean the problem of finding integer solutions to a linear system, that satisfy some constraints. Say, $A{\bf x}^T={\bf b}^T$ such a system, for some integer matrix $A$ and we want to find an integer vector ${\bf x}$ that belongs to a special set ${\mathcal{S}}.$ This problem appears to integer-programming problems, where we have to find some positive integer solutions, so in this case ${\mathcal{S}}={\mathbb{Z}}_{\ge 0}^n.$ Integer Linear Problems have many practical applications, for instance, in {\it{Capital Budgeting}}, {\it{Warehouse Location}} and {\it{Scheduling}}, see \cite[Chapter 9]{ilp}. In the present paper we provide algorithms for attacking compact knapsack problem. Thus, some of the results of this paper may be of independent   interest, not only in cryptography. Without the constraints, systems of linear equations can be solved in polynomial time, either in the real numbers (for instance, using Gauss reduction of the matrix) or over the integers. The latter was showed first time in 1976, using Hermite Normal Form (HNF) of a matrix, by von zur Gathen and Sieveking \cite{von} and Frumken \cite{frumkin}. Furthermore, by Kannan and Bachem was showed a similar result, but using Smith Normal Form (SNF) of a matrix, see \cite{kannan}.

Next, we use this problem to construct a Sigma-id-scheme based on compact knapsack problem. This is an extension of the results of paper \cite{draz2}. What is more, we provide the necessary proofs for completeness, special soundness, and special honest verifier zero knowledge properties, for the general case, i.e. the compact knapsack is not defined by a single equation as in \cite{draz2}, and finally we present the corresponding digital signature by using the Fiat-Shamir transform. We prove that it is secure under the Random Oracle Model (ROM). We remark here that our (computational) assumption is based on the difficulty of solving compact knapsack problem, which it is potentially quantum resistant problem. So, our cryptographic primitives are candidates for the post quantum crypto.

\subsection{Our contribution}
One question that arises from \cite{draz2} is why not to use a linear system instead of one linear equation in order to build our id scheme. For instance, we do not know if the compact knapsack problem becomes harder or easier when we increase the number of equations $m$. This problem is of independent interest in arithmetic linear algebra and complexity theory.  To answer this question, we study the hardness of compact knapsack problem for linear systems using similar methods as in \cite{draz2}. We generalize the attacks and we put a more general framework for the suggested id scheme. Furthermore, we prove that our id scheme is secure under active attacks.

What is more, we extend our study to build digital signatures based on the id scheme using Fiat-Shamir transform. The signature scheme is proved to be secure under active attacks. It is worth mentioning that cryptographic primitives based on compact knapsack problem are suitable for post quantum cryptography, since there is not any quantum polynomial attack, until today, for this problem.
\subsection{Roadmap.}
In section \ref{section-1} we provide the necessary background for lattices. In section \ref{section-2} we give the definition of the compact knapsack problem and we suggest some lattice-based attacks. In the next section we present our id scheme and we build the corresponding digital signature. Additionally, we prove that it is secure under active attacks in the random oracle model. Finally, in the last section we provide some concluding remarks. 
\section{Lattices}\label{section-1}
See \cite{galbraith,mic} for an account on lattices.
\begin{definition}\label{1}
A subset $L\subset {\bf{R}}^n$ is called a lattice if there exist linearly independent set of vectors $B=\{{\bf{b}}_1,{\bf{b}}_2,...,{\bf{b}}_k\}$ of ${\bf{R}}^{n}$ such that
\[L=\Big{\{} \sum_{j=1}^{k}\alpha_j{\bf{b}}_j : \alpha_j\in\mathbb{Z}, 1\leq j\leq k\Big{\}}:=L({\bf b }_1,{\bf b }_2,...,{\bf b }_k).\]
The set $B$ is called a lattice basis of $L.$
\end{definition}
All the bases have the same number of elements, this common number is called {\it{rank}} of the lattice. Let $A$ be the matrix which have as rows the vectors of the basis. Then, we say that the matrix is generated by the rows of $A.$ Below we present the two basic problems on lattices, SVP: Shortest Vector Problem and CVP: Closest Vector Problem.
\subsection{SVP/CVP}
There are two fundamental problems in lattices, the Shortest Vector Problem (SVP), which is the task of finding a shortest vector in a lattice $L(B)$ and the Closest Vector Problem (CVP), where we are looking for a lattice vector closest to another given vector (not in lattice but in the real span of it). 

In a landmark paper of Ajtai, SVP is proved to be NP-hard under randomized reductions \cite{ajtai}. Further, there is the approximation version of SVP, the  ${\text{SVP}}_{\gamma},$ with factor $\gamma\in {\mathbb{R}}_{>1}.$ That is, we are looking for lattice vectors 
${\bf x}\not= {\bf 0}$ with, $||{\bf x}||<\gamma(n)||{\bf y}||$ for every ${\bf y}\in L(B)-\{\bf{0}\}.$ 
A generic algorithm to attack this problem is the enumeration algorithm, see \cite{gama}, which is exponential with respect to the rank of the lattice.

In order to define CVP we need a target vector ${\bf t}$: 
$${\bf t}\in span(B)=\Big{\{} \sum_{j=1}^{k}c_j{\bf{b}}_j : c_j\in{\mathbb{R}} \Big{\}},$$ and we are looking for 
a  vector ${\bf x}\in L(B),$ such that $||{\bf x}-{\bf t}||\leq ||{\bf y}-{\bf t}||$ for every ${\bf y}\in L(B).$ This problem is proved to be NP-hard.

Having a lattice, we need to work with {\it{good}} bases. That is, the basis vectors have small lengths and are almost orthogonal.  A widely known algorithm, that provides such bases is the LLL algorithm, which was developed in 1982 by A.~Lenstra, H.~Lenstra, and L.~Lov$\acute{\text{a}}$sz, see \cite{LLL}.
Furthermore, LLL solves ${\text{SVP}}_{\gamma}$ in polynomial time for $\gamma=2^{k/2},$ where $k$ is the rank of the lattice. In fact, LLL behaves better in practice than in theory. 

\section{Compact knapsack problem}\label{section-2}
\subsection{An integer solution to a linear system}\label{a_solution_to_a_linear_system}
Before we provide our cryptographic system, we present an attack to compact knapsack problem, which is the underlying problem of our scheme.

Let $A=(a_{ij}),$ a $m\times n$ $(m\le n)$ matrix with integer entries. We want to solve the system 
$A{\bf X}^T ={\bf d}^T$, ${\bf X}\in {\mathcal{M}}_{1\times n}({\mathbb{Z}})$  and ${\bf d}\in {\mathcal{M}}_{1\times m}({\mathbb{Z}}).$ Let ${\bf x}$ be an integer solution of the system. We are interested in finding small and balanced\footnote{i.e. the coordinates are of the same magnitude. } solutions ${\bf x}'$s. We follow \cite{Lenstra1}. Let $B$ be the $(n+1)\times (n+m+1)$ matrix,
\begin{equation}\label{definition_of_matrix}
		B=
		\left[\begin{array}{c|c|c}
			I_n &  {\bf 0}_{n\times 1} & N_2A^T  \\
			\hline
			{\bf 0}_{1\times n} &  N_1 & -N_2{\bf d}  
		\end{array}\right],
		\end{equation}
that is,
\begin{equation}\label{definition_of_matrix}
		B=
		\left[\begin{array}{cccc|c|cccc}
			1 &     0 & \dots &   0  &  0  &   N_2a_{11}    &  N_2a_{21}   & \dots & N_2a_{m1}  \\
			0&     1 & \dots &   0    &  0 &   N_2a_{12}    &  N_2a_{22}   & \dots & N_2a_{m2}  \\
			\vdots & \vdots  &  \ddots &  & \vdots &   \vdots & \vdots & \ddots & \vdots   \\
			0 & 0 & \dots & 1        &  0  &    N_2a_{1n}    &  N_2a_{2n}   & \dots & N_2a_{mn}  \\
			\hline
			0 & 0 &  \dots &    0 &  N_1 &  -N_2d_1   &-N_2d_2 & \dots & -N_2d_{n}    \\
			
		\end{array}\right],
		\end{equation}
		where $N_1,N_2,$ are large integers, which we shall choose later. From \cite[Proposition 1]{Lenstra1}, ${\bf x}$ is a solution of $A{\bf X}^T ={\bf d}^T$ if and only if the following two vectors are equal,
		$$
		\left[\begin{array}{ccccc}
			{\bf x} & | & N_1 & | & {\bf 0}_{1\times m}
		\end{array}\right]= 
		\left[\begin{array}{ccc}
			{\bf x} & | & 1
		\end{array}\right] B. $$
		That is, ${\bf x}$ is a solution of $A{\bf X}^T ={\bf d}^T$ if and only if the vector $({\bf x} , N_1 , {\bf 0}_{1\times m})$ belongs to the lattice $L(B)\subset {\mathbb{Z}}^{n+m+1}.$
		Furthermore,  ${\bf x}$ is a solution of $A{\bf X}^T ={\bf 0}_{m\times 1}$ if and only if, 
		$$
		\left[\begin{array}{ccccc}
			{\bf x} & | & 0 & | & {\bf 0}_{1\times m}
		\end{array}\right]= 
		\left[\begin{array}{ccc}
			{\bf x} & | & 0
		\end{array}\right] B. $$
		As previous,  ${\bf x}$ is a solution of $A{\bf X}^T ={\bf 0}_{m\times 1}$ if and only if the vector $({\bf x},0 , {\bf 0}_{1\times m})$ belongs to the lattice $L(B).$
According to \cite[Theorem 4]{Lenstra1} the solution if exists, it is located at the row which contains $N_1$ (according to the form of the matrix, there is exactly one such row). To apply this algorithm we use LLL. So the procedure for finding an integer solution has polynomial time (in practice is very fast).
Finally, we note that, this algorithm in rare cases may fail, even if the system has an integer solution. Then, in practice we choose new values for $N_1$ and $N_2.$
\subsection{CVP-attack to Compact Knapsack Problem}
We generalize the algorithm of \cite[Section 5]{draz2}, for many equations. We set $I_{\alpha}$ be the set of integers having $\alpha-$ bits. Let
$A{\bf X}^T={\bf C}^T,$ where $A\in {\mathcal{M}}_{m\times n}({\mathbb{Z}}),$ ${\bf X}^T\in {\mathcal{M}}_{n\times 1}({\mathbb{Z}})$, ${\bf C}^T\in {\mathcal{M}}_{m\times 1}({\mathbb{Z}})$ and we restrict the solution vector ${\bf X}$ to a set $\mathcal{S}\subset {\mathbb Z}^{n}$. Let $L\subset {\mathbb{Z}}^n$ be the lattice generated by the solutions of the homogeneous system $A{\bf X}^{T}={\bf 0}.$ This lattice has dimension $n-rank(A).$ 
An approach to attack compact knapsack is to reduce it, to a suitable Closest Vector Problem (CVP). The idea is the following. 

 Let ${\bf y}$ be an integer solution of the system $A{\bf X}^T={\bf C}^T$ and choose a suitable target vector ${\bf t}\in span(L)$ (with ${\bf t}\in \mathcal{S}$). 
Then, we solve the CVP instance $CVP(L,{\bf t})$ and say ${\bf b}$ its output. Since vector ${\bf b}$ is close to ${\bf t}$, it is probably in ${\mathcal{S}}.$ So, we expect the solution ${\bf x}={\bf y}+ j{\bf b}$ (for some small integer $j$) to be in ${\mathcal{S}}.$ This can be explained as follows : Since ${\bf y}$ has balanced coordinates (i.e. all the coordinates are close to each other) and assuming that it is short enough, then the sum ${\bf y}+j{\bf b}$ is close to the set ${\mathcal{S}}.$ Thus, we hope that all the entries of  ${\bf y}+j{\bf b}$ for some $j$ are in the set ${\mathcal{S}}.$

We set ${\mathcal{S}}=I_R^n$ and ${\bf t}=(t_{R},...,t_{R})\in {\mathbb{Z}}^n,$ where 
\begin{equation}\label{t_R}
t_{R}\ \text{is\ the\ integer\ } 2^{R-1}+2^{R-2}.
\end{equation}
Using the previous attack, for any $n$ and $m<n$ and any matrix $A,$ we always get a solution in ${\mathcal{S}}$ (here we assume that $m\not =n$ since if $m=n$ the compact knapsack problem is easy). The situation becomes harder if we consider groups of $\{x_j\}_j$'s having different bits. Assume that, $R$ and $n$ are even integers. For instance, if the first $n/2$ entries of ${\bf x}=(x_j)_j$ have $R-$ bits and the other half have $R/2-$ bits, we get a solution having (on average) the half of entries in $I_{R}\cup I_{R/2}.$ We used the target vector, ${\bf t}=(t_{R},..,t_{R},t_{R/2},...,t_{R/2}),$ where the first $n/2$ entries are equal to $t_R$ and the rest to $t_{R/2}.$
\begin{center}
\texttt{CVP-attack}
\end{center}
{\footnotesize{
\noindent
{\bf Input}: \\ $\bullet$ $\mathcal{S}$ a finite subset of ${\mathbb{ Z}}^n,$\\
$\bullet$ $A\in {\mathcal{M}}_{m\times n}({\mathbb{Z}})$ \\
$\bullet$ ${\bf C}$ is an integer vector of ${\mathbb{Z}}^m$ such that the linear system $A{\bf X}^T={\bf C}^T$ has a solution in ${\mathcal{S}}$ \\
$\bullet$ a target vector ${\bf t}\in {\bf R}^{n}\cap {\mathcal{S}}$ and \\ 
$\bullet$ a positive integer $\alpha$.
\ \\
{\bf Output}: \\A solution ${\bf x}\in {\mathcal{S}}$ such that $A{\bf x}^{T}={\bf C}^T$ or in the worst case returns a solution that satisfies some constraints (i.e. some coordinates are in ${\mathcal{S}}$) \\\\
	\texttt{01.}   {\texttt{compute a solution ${\bf y}$ of $A{\bf Y}^T={\bf C}^T$}}\ (see subsection \ref{a_solution_to_a_linear_system})\\
	    \texttt{02.}   {\texttt{compute a basis $B$ of the lattice 
    $L=\{ {\bf x}\in{\mathbb Z}^{n}: A{\bf x}^T={\bf 0} \}$ }}\\
    \texttt{03.}   $B\leftarrow LLL(B)$\\
	\texttt{04.}   {\texttt{${\bf b}\leftarrow 
	CVP( L(B),{\bf t})$}}\\
	\texttt{05.}   {\texttt{return the best vector (i.e. the one that meets more constraints) from the set $$\{{\bf y}+j{\bf b}:j=-\alpha,\dots ,\alpha\}$$}}
}}
	In line 01, we use the results of subsection \ref{a_solution_to_a_linear_system} in order to compute an integer solution of the system $A{\bf Y}^T={\bf C}^T.$ In the next line one way to compute a basis is to use Smith Normal Form (SNF) of the matrix $A$. In fact, if $S$ is the SNF of $A,$ then there are $P,Q$ unimodular matrices such that $S=PAQ.$ Then it is proved that, the last $n-rank(A)$ columns of $Q$ is a basis of the lattice $AX={\bf 0}$ (see the Appendix). In the next line, 03, we compute the LLL of the basis $B$ of $L.$ The complexity of this step is polynomial, since LLL is polynomial and floating point versions are very fast in practice. In line 04, we solve CVP for the lattice $L.$ In practice, we use Babai algorithm \cite[Chapter 18]{galbraith}. 		
\subsubsection{Experiments}  We provide some experiments\footnote{For the code see \url{https://github.com/drazioti/compact-knapsack}} based on the previous attack. Let $R$ and $n$ be a positive integers. We set,
\begin{equation}\label{set_S}
\mathcal{S}_{\alpha_1,\alpha_2,...,\alpha_k}(n,R) =I_{\alpha_1\cdot R}^{n/k} \times I_{\alpha_2\cdot R}^{n/k}  \times \cdots \times I_{\alpha_k\cdot R}^{n/k}\subset {\mathbb{Z}}^n,\ \text{for\ some}\ k|n, 
\end{equation}
and $\alpha_j>0\ (j=1,2,...,k),$ such that $\alpha_jR\in {\mathbb{Z}}.$ For instance,
 ${\mathcal{S}}_{\alpha,\beta}(n,R)=I_{\alpha\cdot R}^{n/2}\times I_{\beta\cdot R}^{n/2}$ and 
 $$\mathcal{S}_1(n,R)=I_R^{n}=\{{\bf x}=(x_i)\in {\mathbb{Z}}^n: 2^{R-1}\le x_i\le 2^{R}-1\}.$$ 
 As previous $m$ is the number of equations and $n$ the number of unknowns. If ${\bf x}\in \mathcal{S}_{\alpha_1,\alpha_2,...,\alpha_k}(n,R),$ then 
$${\bf t}=(t_{\alpha_1\cdot R},...,t_{\alpha_1\cdot R};t_{\alpha_2\cdot R},...,t_{\alpha_2\cdot R};...;t_{\alpha_k\cdot R},...,t_{\alpha_k\cdot R}),$$
where $t_{sR}$ was defined in (\ref{t_R}).
The distance between two consecutive \texttt{;} is $n/k-$entries.
\begin{table}[h]
\begin{center}
\begin{tabular}{|c|c|c|c|c|c|c|} 
\hline
$ $ & $n=50 $ & $n=50 $ & $n=50 $ & $n=48 $ & $n=48 $ & $n=50 $\\
\hline
$m$ &  $\mathcal{S}_1$ &  $\mathcal{S}_{1,1/2}$ & $\mathcal{S}_{\frac{1}{2},\frac{1}{4}}$ & $\mathcal{S}_{\frac{1}{2},\frac{1}{4},\frac{1}{8}}$  & $\mathcal{S}_{1,\frac{1}{2},\frac{1}{4},\frac{1}{8}}$ & $\mathcal{S}_{\big(1,\frac{1}{2},\frac{1}{4},\frac{1}{8},\frac{1}{16}\big)}$     		 \\ 
\hline  \hline
 $1$ &   $100\% $ & $50\%$ & $50\%$	& $63.3\%$ & $ 71.25\% $ & $ 75.6\% $\\ 
\hline
  $2$ &  $100\% $ & $50\%$ & $50\%$& $58.8\%$ & $ 59.4\%$  & $65.6\% $ \\ 
\hline
  $10$ & $100\%$ & $50\%$ & $50\%$	& $35.8\%$ & $34.3\% $  & $  32.6\%$ \\ 
\hline
  $20$ & $100\% $ & $50\%$ & $50\%$	& $33.3\% $ & $26.6\% $  & $  24.8\%$\\ 
\hline
  $32$ & $100\% $ & $50\%$ & $50\% $	& $33.3\% $ & $26.45\%$  & $ 24\% $\\ 
\hline
  $40$ & $100\% $ & $50\%$ & \textcolor{red}{$52\%$} & $100\% $ & $27.7\% $	 & $  23.2\%$  \\ 
\hline
\end{tabular}
\end{center}
\caption{For the second, third, fourth and last column $R=80$, $n=50,$ and $A\xleftarrow{\$} M_n(I_{R/8}).$ We executed 20 random instances for each row. The red text indicates that, whereas we have on average $<100\%$ successes, the attack found at least one solution. For the fifth and sixth columns $R=80$ and $n=48$. For all the examples we pick $\alpha=10$.}\label{Table:1}
\end{table} 

\begin{table}[h]
\begin{center}
\begin{tabular}{|c|c|c|c|} 
\hline
$ $ & $n=54$ & $n=70$ & $n=64$ \\
$ $ & $R=64$ & $R=128$ & $R=256$ \\
\hline
$m$ &  $\mathcal{S}_{\big(1,\frac{1}{2},\frac{1}{4},\frac{1}{8},\frac{1}{16},\frac{1}{32}\big)}$ &  $\mathcal{S}_{\big(1,\frac{1}{2},\frac{1}{4},\frac{1}{8},\frac{1}{16},\frac{1}{32},\frac{1}{64}\big)}$ & $\mathcal{S}_{\big(1,\frac{1}{2},\frac{1}{4},\frac{1}{8},\frac{1}{16},\frac{1}{32},\frac{1}{64},\frac{1}{128}\big)}$     		 \\ 
\hline  \hline
 $1$ &   $78.7\%$ & $82\%$ & $84.4\%$ \\ 
\hline
  $2$ &  $65.7\% $ & $67\%$ & $72.1\%$ \\ 
\hline
  $10$ & $31.2\%$ & $34.9\%$ & $33.5\%$\\ 
\hline
  $20$ & $24.1\%$ & $26.4\%$ & $28.5\%$ \\ 
\hline
  $32$ & $22.1\% $ & $24.3\% $ & $25.5\% $  \\ 
\hline
  $40$ & $22.9\% $ & $22\%$ & $23.8\%$\\ 
\hline
\end{tabular}
\end{center}
\caption{$A\xleftarrow{\$} M_n(I_{R/8}).$ We executed $50$ random instances for each row.}\label{Table:2}
\end{table} 

We remark (see Table \ref{Table:1}) that the problem is slightly easier for some solution spaces (columns 2,3,4) and for some other spaces is harder, when we increase the number of equations $m$ (and keeping the same $n,R$). Finally, from the previous two tables we conclude that all the solution spaces provide evidences that compact knapsack is harder if we consider a (non trivial) system than a single linear equation.

\subsection{An improvement of the previous algorithm}
We notice from the previous experiments that it is easy to solve compact knapsack when the solution space is $I_R^n.$ I.e. we pick ${\bf x}$ having each entry exactly $R$ bits. So, we can try a {\it{divide and conquer}} approach. We can create sub tasks where each one solves a compact knapsack problem to the $t-$ bits space (for some finite values of $t$). Finally, we will merge all such sub solutions in order to get a solution of the initial system.
We suppose that the solution space is of the form 
$\mathcal{S}_{\alpha_1,\alpha_2,...,\alpha_k}(n).$ Then we split the 
initial matrix $A\in {\mathcal{M}}_{m\times n}$ to 
$[A_1|A_2|\cdots|A_k],$ where $A_i$ is $m\times n/k.$ Then, we randomly pick some integer vector ${\bf c}_i$ $(m\times 1)$ and we are looking for solutions in $I_{\alpha_i R}^n$ for the system $A_iX={\bf c}_i.$ Now, if the random choice is {\it{good}} enough, then there is room for improvement in the previous attack. We provide the pseudocode\footnote{An implementation in Sagemath can be found in : \url{https://github.com/drazioti/compact-knapsack}}.

\begin{center}
\texttt{Divide and Conquer Attack}
\end{center}
{\footnotesize{
\noindent
{\bf Input}:\\
$\bullet$ $R$ is a positive integer\\
 $\bullet$ ${\mathcal{S}}=\mathcal{S}_{\alpha_1,\alpha_2,...,\alpha_k}$ a finite subset of ${\mathbb{ Z}}^n,$\\
$\bullet$ $A\in {\mathcal{M}}_{m\times n}({\mathbb{Z}})$ \\
$\bullet$ ${\bf C}$ is an integer vector of ${\mathbb{Z}}^m$ such that the linear system $A{\bf X}^T={\bf C}^T$ has a solution in ${\mathcal{S}}$  \\ 
$\bullet$ $\alpha$ a positive integer\\
$\bullet$ $(\beta_2,...,\beta_k)$ some positive integers such that $\beta_i R$ is integer\\
\ \\
{\bf Output}: \\ A solution ${\bf x}\in {\mathcal{S}}$ such that $A{\bf x}^{T}={\bf C}^T$ or in the worst case returns a solution that satisfies some constraints (i.e. some coordinates are in ${\mathcal{S}}$)\\\\
    \texttt{01.} {\texttt{Initialize a list ${\mathcal{N}}=[]$}}\\
    \texttt{02.} ${\bf{h}}_i^T \xleftarrow{\$} I_{\beta_i R}^{m}$ {\texttt{for $i=2,...,k$}{\texttt{ \# ${\bf{h}}_i^T$ are column matrices with $m$ entries. $\beta_i$ may be chosen equal to $\alpha_i$ but it is not necessary.}}}\\
    \texttt{03.} {\texttt{${\bf h}_1^T \leftarrow {\bf C}^T- \sum_{j=2}^{k}{\bf {h}}_j^{T}$}}\\
    \texttt{04.} {\texttt{Split $A=[A_1|A_2|...|A_k]$ to blocks $A_i,$ such that $A_i \in {\mathcal{M}}_{m\times \frac{n}{k}}({\mathbb{Z}})$}\\
    \texttt{05.}   ${\bf t}_i\leftarrow (t_{\alpha_i\cdot R},...,t_{\alpha_i\cdot R})$ for $i=1,2,...,k$} {\texttt{\# ${\bf t}_i$'s are the target vectors having $n/k$ entries }\\
    \texttt{06.}   {\texttt{For $i$ in $\{1,2,...,k\}$}}\\
	\texttt{07.}    \hspace{0.7cm} {\texttt{compute a solution ${\bf y}_i\in{\mathbb Z}^{n/k}$ of $A_i{\bf Y}^T={\bf h}_i^T$}}\\
	\texttt{08.} \hspace{0.7cm}  \texttt{compute a basis $B_i$ of the lattice 
    $L_i=\{ {\bf x}\in{\mathbb Z}^{n/k}: A_i{\bf x}^T={\bf 0} \}$}\\
	\texttt{09.} \hspace{0.7cm}  \texttt{${\bf b}_i\leftarrow 
	CVP( L(B_i),{\bf t}_i)$}\\
	\texttt{10.}    \hspace{0.7cm} \texttt{compute the best vector ${\bf x}_i$ \# i.e. the one that meets more constraints from the set $\{{\bf y}_i+j{\bf b}_i:j=-\alpha,\dots ,\alpha\}$}\\
	\texttt{11.}    \hspace{0.7cm} ${\mathcal{N}}\leftarrow \mathcal{N}\cup \{{\bf x}_i\}$ \\
	\texttt{12.} ${\bf x}\leftarrow ({\bf x}_1,{\bf x}_2,...,{\bf x}_k)$, where ${\mathcal{N}}=\{{\bf x}_1,...,{\bf x}_k\}$\\
	\texttt{13.}   return ${\bf x}$
}}}

\begin{table}[h!tb]
\begin{center}
\begin{tabular}{|c|c|c|c|} 
\hline
$ $ &  $n=50 $ & $n=50 $ & $n=48 $\\
\hline
$m$ &   $\mathcal{S}_{1,1/2}$ & $\mathcal{S}_{\frac{1}{2},\frac{1}{4}}$ & $\mathcal{S}_{\frac{1}{2},\frac{1}{4},\frac{1}{8}}$      		 \\ 
\hline  \hline
 $1$ &    \textcolor{brown}{$50\%$}-\textcolor{magenta}{$64.5\%$} & \textcolor{brown}{$50\%$}-\textcolor{magenta}{$64.5\%$}	& \textcolor{brown}{$63.3\%$}-\textcolor{magenta}{$55\%$} \\ 
\hline
  $2$ &  \textcolor{brown}{$50\%$}-\textcolor{magenta}{$62.5\%$} & \textcolor{brown}{$50\%$}-\textcolor{magenta}{$64\%$} & \textcolor{brown}{$58.8\%$}-\textcolor{magenta}{$52\%$} \\ 
\hline
  $10$ & \textcolor{brown}{$50\%$}-\textcolor{magenta}{$62.2\%$} & \textcolor{brown}{$50\%$}-\textcolor{magenta}{$63\%$} & \textcolor{brown}{$35.8\%$}-\textcolor{magenta}{$43\%$}  \\ 
\hline
  $20$ &  \textcolor{brown}{$50\%$}-\textcolor{magenta}{$52\%$} & \textcolor{brown}{$50\%$} -\textcolor{magenta}{$50\%$}	& \textcolor{brown}{$33.3\%$}-\textcolor{magenta}{$33\%$}  \\ 
\hline
  $32$ &  \textcolor{brown}{$50\%$}-\textcolor{magenta}{$50\%$}& \textcolor{brown}{$50\%$}-\textcolor{magenta}{$50\%$}	& \textcolor{brown}{$33.3\%$}-\textcolor{magenta}{$33\%$} \\ 
\hline
  $40$ &  \textcolor{brown}{$50\%$}-\textcolor{magenta}{$50\%$} & \textcolor{brown}{$50\%$}-\textcolor{magenta}{$50\%$} & \textcolor{brown}{$100\%$}-\textcolor{magenta}{$33\%$}   \\ 
\hline
\end{tabular}
\end{center}
\caption{With the brown color we set the same data as in Table \ref{Table:1}.	With magenta we indicate the results of {\it{Divide and Conquer}} Attack. We executed 20 random instances for each row. For all the examples we chose $\alpha=10$ and $\beta_i=\alpha_i$, $i\ge 2$. For  $\mathcal{S}_{1,\frac{1}{2},\frac{1}{4},\frac{1}{8}}$ and $\mathcal{S}_{\big(1,\frac{1}{2},\frac{1}{4},\frac{1}{8},\frac{1}{16}\big)}$ we did not find any improvement.} \label{Table:3}
\end{table}

\section{A three move id-scheme based on compact knapsack}\label{section-4}

The Id-scheme that we propose, based on compact knapsack consists of three moves. Suppose that Alice (the prover) holds a secret key and sends a message to Bob (the verifier). We call this step,  \textit{commitment}. Bob responds with a random string, the \textit{challenge} (or \textit{exam}) and Alice provides a \textit{response}. In the last step, Bob applies a verification algorithm which has as inputs the public key of Alice and the previous conservation, in order to decide if he will accept or reject Alice's id. The length of the challenge is the security parameter. Our aim is to provide a proof of knowledge for the compact knapsack problem.

Let $A$ be an $m\times n$ matrix and we consider the linear system $ AX={\bf b}^T.$ Let ${\bf x}$ belongs to a set 	${\mathcal{S}}\subset {\mathbb{Z}}^n.$ We assume that compact knapsack is difficult to be solved in ${\mathcal{S}}.$ 
Furthermore, we choose another set ${\mathcal{S}}'\subset {\mathbb{Z}}^n,$ where the compact knapsack problem is also hard on ${\mathcal{S}}'.$ What is more, we make the following assumption\footnote{We shall clarify this assumption in subsection (\ref{subsub1}).}:
\begin{equation}\label{assumption}
Pr( ({\bf x},{\bf k})\xleftarrow{\$} {\mathcal{S}}\times {\mathcal{S'}} : {\bf x} + {\bf k}\in {\mathcal{S}} )\approx 1.
\end{equation}
The quadruple $(A,{\bf b},{\mathcal{S}},{\mathcal{S}}')$ is public.  Therefore, the previous generation algorithm $G$, on input a random seed generates a public key $pk=(A,{\bf b},{\mathcal{S}},{\mathcal{S}}')$ and a secret key ${\bf x}\in {\mathcal{S}},$ such that $A{\bf x}^T={\bf b}^T.$ Since, we assumed that compact knapsack is difficult for the triple $(A,{\bf b},{\mathcal{S}}),$ we say that $G$ is one way (this is the definition 19.6 of \cite{BonehShoup}).

The following id-scheme is repeated $t$-times. The challenge space ${\mathcal{C}}$ is of the form $\{0,1\}^t,$ where $t$ is the security parameter and since $1/|{\mathcal{C}}|$ is negligible\footnote{A function $f:{\mathbb{N}}\rightarrow {\bf {R}}$ is called negligible if and only if $|f(n)|
<1/n^c$ for large $n$ and some positive constant $c.$ We also call such a function $f$ super-poly.} for large $t$ we say that ${\mathcal{C}}$ is large. 

\begin{itemize}
    \item Alice picks a random vector $\textbf{k} \in {\mathcal{S}}'$. Then, she computes
    $$ A \cdot \textbf{k}^T = \textbf{r}^T$$
    and sends it to Bob (\textit{commitment}).
    \item Bob picks a random bit $e$ and sends it to Alice (\textit{challenge}).
    \item Alice computes
    $$ \textbf{s} = \textbf{k} + e \textbf{x}$$
    and sends $\textbf{s}$ to Bob (\textit{response}).
    \item Bob verifies the equality $A \cdot \textbf{s}^T = (\textbf{r} + e \textbf{b})^T$ and that $\textbf{s} \in {\mathcal{S}}'$ if $e=0$. Now, if $e=1$, our assumption (\ref{assumption}) allows Alice to choose from the beginning the set ${\mathcal{S}}$ and ${\mathcal{S}}'$, such that $\textbf{s} \in {\mathcal{S}}$ with large probability $(\approx 1)$.
\end{itemize}
See also Fig. \ref{id_scheme}.\\
\noindent \textit{Proof of correctness.}
$$ A \cdot \textbf{s}^T = A \cdot \textbf{k}^T + e A \cdot \textbf{x}^T = (\textbf{r} + e \textbf{b})^T$$
Also, $\textbf{s}$ satisfies the constraints of the scheme. Indeed, if $e=0,$ then ${\bf s}={\bf k}\in {\mathcal{S}'},$ and if $e=1$ then ${\bf s}={\bf k}+{\bf x}$ which belongs to ${\mathcal{S}}$ with large probability (this is from assumption (\ref{assumption})).
\begin{figure}[h!tb]
\centering
\rotatebox{0}{\scalebox{0.37}{\includegraphics{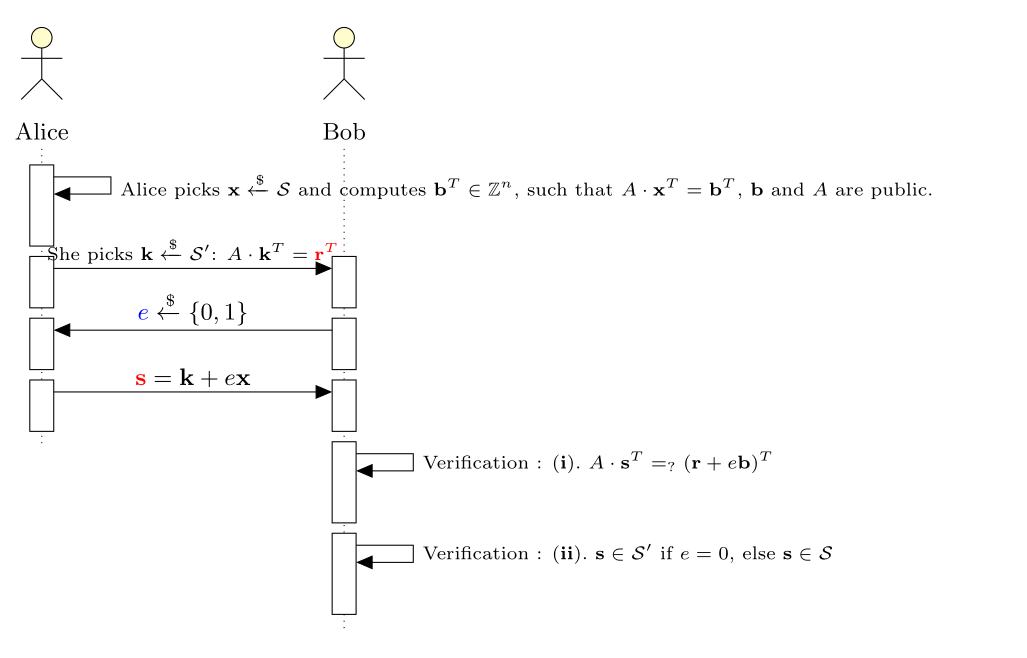}}}
\caption{Compact Knapsack ID-Scheme. Alice sends the computed values in red to Bob. Bob does not know ${\bf k}$ and ${\bf x}.$}
\label{id_scheme}
\end{figure}
%\begin{remark}
%If it may happens ${\mathcal{S}}'\subset {\mathcal{S}}$ (or the inverse), then the verification step is simplified, since whatever $e$ is chosen, we always have ${\bf s}\in {\mathcal{S}}.$ 
%\end{remark}
\subsubsection{How to choose ${\mathcal{S}}$ and ${\mathcal{S}}'$?}\label{subsub1}
We shall choose  ${\mathcal{S}}$ and ${\mathcal{S}}'$. Both sets are defined according to relation (\ref{set_S}) i.e.
${\mathcal{S}}_{\alpha_1,\alpha_2,\dots,\alpha_k}(n,R)$. We start with some auxiliary lemmata.
\begin{lemma}\label{lemma}
Let $0<a<b<c<d$ integers and $N_1=[a,b]\cap {\mathbb{Z}},\ N_2=[c,d]\cap {\mathbb{Z}}.$ Suppose that $b\leq d-c+1.$ Let also $x_1,x_2$ are chosen uniformly from $N_1,N_2$, respectively.
Then, 
$$Pr(x_1+x_2\not \in N_2) = \frac{(a+b)}{2(d-c+1)}.$$
\end{lemma}
\begin{proof}
We follow \cite[Lemma 6.2]{draz2}. First we calculate the number of pairs $(x_1, x_2)\in N_1\times N_2$ such that, $x_1 + x_2 > d.$ We fix for a moment $x_1=a,$ then in order to have $x_1+x_2>d$ or $x_2>d-a$ we count the $x_2'$s in $N_2$ that satisfies $x_2>d-a.$ So we start from $d-a+1$ and we count until $d.$ Overall, $d-(d-a+1) + 1= a.$ Similar if $x_1=a+1$ we get $a+1$ possible $x_2$'s and so on, overall we get 
$$a+(a+1)+\cdots +b = \frac{a+b}{2}.$$  The number of elements in $N_2$ are $d-c+1,$ so the Lemma follows.
\end{proof}
\begin{corollary}
Suppose that $\beta\leq \frac{1}{R\ln{2}} \ln(2^{\alpha R-1} +1).$ Let also, $x_1\xleftarrow{\$} {\mathcal{S}}_{\beta}(1,R)$ and $x_2\xleftarrow{\$} {\mathcal{S}}_{\alpha}(1,R)$ where $\alpha,\beta<1$ positive real numbers such that $\alpha R,\beta R $ are positive integers.
Then, $$Pr(x_1+x_2\not \in {\mathcal{S}}_{\alpha})=\frac{3\cdot 2^{\beta R-1}-1}{2^{\alpha R}}.$$
\end{corollary}
\begin{proof}
Since ${\mathcal{S}}_{\alpha}(1,R) = I_{\alpha R}=N_2$ and ${\mathcal{S}}_{\beta}(1,R)=I_{\beta R}=N_1,$ we set $a=2^{\beta R -1}, b=2^{\beta R} -1,$ $c=2^{\alpha R -1}, d=2^{\alpha R} -1$ and we apply Lemma \ref{lemma}. 
Since, 
$$b\le d-c+1 \ \text{or}\ 2^{\beta R}-1\le 2^{\alpha R}-2^{\alpha R-1}\ \text{or}\ \beta\leq \frac{1}{R\ln{2}} \ln(2^{\alpha R-1} +1),$$ the hypothesis of the previous Lemma is satisfied. Thus,
$$Pr(x_1+x_2\not \in {\mathcal{S}}_{\alpha})=\frac{2^{\beta R-1}+2^{\beta R}-1}{2^{\alpha R + 1}-2^{\alpha R}}.$$
The result follows.
\end{proof}

\begin{corollary}
Suppose that $\beta\leq \frac{1}{R\ln{2}} \ln(2^{\alpha R-1} +1),$ where $\alpha,\beta<1$ positive real numbers such that $\alpha R$ and $\beta R $ are positive integers. Let also, ${\bf x}_1\xleftarrow{\$} {\mathcal{S}}_{\beta }(n,R)=I_{\beta R}^n$ and ${\bf x}_2\xleftarrow{\$} {\mathcal{S}}_{\alpha}(n,R)=I_{\alpha R}^n.$
Then, $$Pr({\bf x}_1 + {\bf x}_2\in {\mathcal{S}}_{\alpha}(n))=\Big(1-\frac{3\cdot 2^{\beta R-1}-1}{2^{\alpha R}}\Big)^n.$$
\end{corollary}
\begin{proof}
Easy to check from the previous Corollary.
\end{proof}

\begin{lemma}\label{second_lemma}
We set 	${\mathcal{S}}={\mathcal{S}}_{\alpha_1,\alpha_2,\dots,\alpha_k}(n,R)$ and ${\mathcal{S}}'={\mathcal{S}}_{\beta_1,\beta_2,\dots,\beta_k}(n,R),$ where $$\beta_i\leq \frac{1}{R\ln{2}} \ln(2^{\alpha_i R-1} +1)\ (i=1,2,...,k).$$ Then,
$$Pr\{({\bf x},{\bf k})\xleftarrow{\$} {\mathcal{S}}\times {\mathcal{S}}': {\bf x} + {\bf k} \in {\mathcal{S}}\} =\prod_{i=1}^{k} \Big(1-\frac{3\cdot 2^{\beta_i R-1}-1}{2^{\alpha_i R}}\Big)^{n/k}.$$
\end{lemma}
\begin{proof}
Since, 
$$\beta_i\leq \frac{1}{R\ln{2}} \ln(2^{\alpha_i R-1} +1),\ \text{for}\ i=1,2,...,k$$ from the previous Corollary we get
$ Pr\big\{({\bf x},{\bf k})\xleftarrow{\$} {\mathcal{S}}\times {\mathcal{S}}': {\bf x} + {\bf k}\in {\mathcal{S}}\big\}=$\
$$\prod_{i=1}^k Pr\Big\{({\bf x}_i,{\bf k}_i)\xleftarrow{\$} {\mathcal{S}}_{\alpha_i}(n/k,R)\times {\mathcal{S}}_{\beta_i}(n/k,R): {\bf x}_i + {\bf k}_i \in {\mathcal{S}}_{\alpha_i}(n/k,R)\Big\}.$$
Indeed, 
${\bf x} + {\bf k}\in {\mathcal{S}}\ \text{for } ({\bf x},{\bf k})\in{\mathcal{S}}\times {\mathcal{S}}' \ \textbf{\ if-f\ }$ 
$${\bf x}_i + {\bf k}_i \in {\mathcal{S}}_{\alpha_i}(n/k,R) \text{\ for\ all } ({\bf x}_i,{\bf k}_i)\in \mathcal{S}_{\alpha_i}(n/k,R)\times {\mathcal{S}}_{\beta_i}(n/k,R).$$
The Lemma follows.
\end{proof}

Now we have to choose $\beta_i, \alpha_i$ such that the relation (\ref{assumption}) is satisfied. 
Say $\varepsilon$ is a small positive real number, close to $0.$ Then, we set
\begin{equation}\label{probability}
\frac{3\cdot 2^{\beta_i R-1}-1}{2^{\alpha_i R}} = \varepsilon_i.
\end{equation}
After some calculations we get $:$
$$ \beta_i = \frac{1}{R\ln{2}}\ln \frac{\varepsilon_i 2^{\alpha_i R+1}+2}{3}.$$
Since, we want also $\beta_i\le \frac{1}{R\ln{2}} \ln(2^{\alpha_i R-1} +1)$ we choose $\varepsilon$ small enough in order to satisfy the previous inequality\footnote{It is enough to choose $\varepsilon_i<1/2.$}. 
Then, according to Lemma \ref{second_lemma}, 
\begin{equation}\label{basic_probability}
Pr\{({\bf x},{\bf k})\xleftarrow{\$} {\mathcal{S}}\times {\mathcal{S}}': {\bf x} + {\bf k} \in {\mathcal{S}}\}=(1-\varepsilon_1)^{n/k}\cdots (1-\varepsilon_k)^{n/k}\approx 1.
\end{equation}

\subsection{Sigma-Id-Schemes}
In cryptography we are interested in Sigma protocols, which are id schemes with some specific properties. We want our id scheme be :\\
$({\bf i}).$ Complete\footnote{This is in fact the {\it{proof of correctness}}}, $({\bf ii}).$ Special sound\footnote{in \cite[Section 19.4.1, p. 742]{BonehShoup}, this is called {\it{Knowledge soundness}}.}, and $({\bf iii}).$ Special Honest Verifier Zero Knowledge (S-HVZK). 

%$({\bf iii}).$ uniform over the {\texttt{responses}} space.\\
We will explain one by one what the previous statements mean, and we shall prove that our id scheme is a Sigma protocol.\\
$({\bf i}).$ The first basic property of a Sigma-id-scheme is {\it{completeness}}, i.e. an honest prover can always convince a verifier with some large probability $1-\alpha$. We call $\alpha$ {\it{completeness error}}. Using relation (\ref{basic_probability}) we can choose $\alpha$ be very small, in fact we can choose it $\alpha = 1 -n\varepsilon +O(\varepsilon^2)$ (we assumed that $\varepsilon_1=\cdots=\varepsilon_k=\varepsilon$). So for small $\varepsilon$, the verifier accepts with high probability.\\
$({\bf ii}).$ The second basic property is the {\it{special soundness}}.
Special sound means that, it is hard to compute two valid transcripts\footnote{I.e. they pass the verification test.} 
$({\texttt{commit}}, {\texttt{challenge}}, {\texttt{response}})$ such that the {\texttt{commitments}} are the same and the {\texttt{challenges}} are different. We shall show that this property is valid, in our id scheme, if compact knapsack is hard.
 We assume that ${\mathcal{S}},{\mathcal{S}}'$ are chosen as in (\ref{basic_probability}). Also, we set $\Sigma_{\bf b} = \{{\bf x}\in {\mathcal{S}}:A{\bf x}^T={\bf b}^T\}.$
\begin{proposition}
Let ${\mathcal{T}}_1,{\mathcal{T}}_2$ be two valid transcripts, of the previous id scheme, such that $\big( ({\bf r}_i), {\bf e}, ({\bf s}_i\big))$, $\big( ({\bf r}_i), {\bf e}', ({\bf s}'_i)\big)$ (resp.), $1\leq i\leq t,$ where $t$ is the number of iterations of the protocol. Let also ${\bf e}\not={\bf e}'.$ Then, we can efficiently find an element of ${\Sigma_{\bf b}},$ with high probability.
\end{proposition}
\begin{proof}
Let $j$ such that $e_j=0$ and $e'_j=1.$ Then, $A{\bf s}_j^{T}={\bf r}_j^T$ and $A{\bf s}_j'^{T}={\bf r}_j^T+{\bf b}^T,$ where ${\bf s}_j={\bf k}_j$ and ${\bf s}'_j={\bf k}'_j+{\bf x}$ belongs to ${\mathcal{S}}$ with high probability. Thus, ${\bf s}_j'-{\bf s}_j={\bf x}\in {\mathcal{S}}$ 	with high probability. Since, $A({\bf s}_j'^T-{\bf s}_j^T)={\bf b}^T,$ i.e. ${\bf s}_j'-{\bf s}_j\in \Sigma_{\bf b}$. the Proposition follows.
\end{proof}
Since we have assumed that compact knapsack problem is difficult for the specific sets ${\mathcal{S}}$ and ${\mathcal{S}'},$ we conclude that the system is special sound. Furthermore, since the challenge space is also large, the system is proof of knowledge for the compact knapsack.\\
$({\bf iii}).$ The third property is the {\it{special HVZK}}, which means that, given any challenge $e,$ we can simulate a transcript  $({\bf r}, e, {\bf s})$, which is indistinguishable from a real transcript with challenge $e.$ Real transcript is the transcript generated by an honest prover and an honest verifier. The simulator takes as input the public key and a challenge $e,$ and it creates ${\bf s}$ according to $e,$ i.e. the simulator chooses ${\bf s}$ uniformly from  the set ${\mathcal{S}}'$ if $e=0,$ else it picks ${\bf s}$ from ${\mathcal{S}}.$ Finally, we set ${\bf r}\leftarrow A{\bf s}^T - e{\bf b}^T.$ Therefore, the transcript $({\bf r},e,{\bf s})$ follows the same distribution as the transcript between the real prover and the real verifier.

\subsection{Soundness}
Another useful requirement is the {\it{soundness}} property.  
Let Eve be an adversary. The scheme is sound if Eve knowing only the public key, can pass the verification test with only negligible probability.
The soundness of the scheme depends on the number of iterations $t.$ For  $t=1$ the protocol is not sound, since Eve with probability 1/2 can pass the verification test.
 
Indeed, say that Eve by tossing up a fair coin, picks the right $e'\in\{0,1\}.$ Then, she computes a random vector ${\bf s}\in {\mathcal{S}}$ if $e'=0,$ else she chooses a random ${\bf s}$ from ${\mathcal{S}'}.$ 
Then, the pair $({\bf r}^T=A\cdot {\bf s}^T,{\bf s})$ passes the verification step if $e'=0$, else the pair
$({\bf r}^T=A{\bf s}^T-{\bf b}^T,{\bf s})$ passes the verification step.  Thus, Eve can pass the verification test with probability $1/2.$

Since the protocol is not one round, but we execute $t-$ rounds, therefore, for $t=80$ the success rate of Eve is $2^{-80}.$
To prove that the scheme is sound\footnote{In \cite{Damgard} {\it{sound}} protocols are called {\it{knowledge sound}} protocols, see \cite[Definition 2]{Damgard}. In \cite{BonehShoup} is called {\it{secure under direct attacks}} and the notion of {\it{knowledge soundness}} coincides with the {\it{special soundness}}. Here we follow the terminology of \cite{BonehShoup}. However, since we report also some results from \cite{Damgard} we point out the differences.} we have to show that this success rate can not be improved unless the compact knapsack problem is easy.
Our system is sound, because is a Sigma protocol with large challenge space. This is proved in \cite[Theorem 19.14]{BonehShoup}. A similar proof is provided in \cite[Theorem 1]{Damgard}.

\subsection{Reducing the information complexity}
We can seemingly change our Sigma protocol, in such a way, that the number of rounds be inserted in each step of the protocol. I.e. the prover generates vectors ${\bf k}_1,...,{\bf k}_t$ and the commitments form a matrix ${\mathcal{R}}$ with columns ${\bf r}_1,...,{\bf r}_t$, where ${\bf r}_i^{T}=A{\bf k}_i^T.$ Similarly the verifier sends a binary vector ${\bf e}$ of $\{0,1\}^t$ and finally the prover, in the third step, sends a matrix ${\mathbb{S}}$ containing as columns the vectors ${\bf s}_1={\bf k}_1+e_1{\bf x},...,{\bf s}_t={\bf k}_t+e_t{\bf x}.$ Then, the new id-scheme generated by parallel repetition, it is again Sigma protocol as the initial one. For instance, see \cite[Ex. 19.5, p.766]{BonehShoup}. So, our system can be described as follows:
\begin{itemize}
    \item (\textit{Commitment}). Alice picks a $n\times t$ random matrix $\mathcal{K}$  with columns, ${\bf k}_1,...,{\bf k}_t$ in  ${\mathcal{S}}'$. I.e.
    $${\mathcal{K}}=\left[\begin{array}{ccc}
			| &    &   | \\
			{\bf k}_1^T &  \dots &  {\bf k}_t^T  \\
			| &  &|  		
		\end{array}\right].$$    
     Then, she computes the matrix:
  $$  {\mathcal{R}} = A  \mathcal{K} =\left[\begin{array}{ccc}
			| &    &   | \\
			A{\bf k}_1^T &  \dots &  A{\bf k}_t^T  \\
			| &  &|  \	
		\end{array}\right]\in {\mathbb{Z}}^{m\times t},$$
    and sends it to Bob.
    \item (\textit{Challenge}). Bob randomly picks a binary vector ${\bf e}=(e_1,...,e_t)$ from ${\mathcal{C}}=\{0,1\}^t$ (challenge space) and sends it to Alice.
    \item (\textit{Response}). Alice computes
    $$ \textbf{s}_i = \textbf{k}_i + e_i \textbf{x}\in {\mathbb{Z}}^n$$
    and sends 
$${\mathbb{S}}=\left[\begin{array}{ccc}
			| &    &   | \\
			{\bf s}_1^T &  \dots &  {\bf s}_t^T  \\
			| &  &|  
		\end{array}\right] ={\mathcal{K}}+\left[\begin{array}{ccc}
			| &    &   | \\
			e_1{\bf x}^T &  \dots &  e_t{\bf x}^T  \\
			| &  &|  
		\end{array}\right]\in {\mathbb{Z}}^{n\times t},$$    
    to Bob.
    \item (\textit{Verification}). We set ${\bf r}_i^T=A{\bf k}_i^T \in {\mathbb{Z}}^m, 1\le i\le t.$ Bob computes $A{\mathbb{S}}.$ Then, he verifies the equality of matrices: 
    $$\left[\begin{array}{ccc}
			| &    &   | \\
			A{\bf s}_1^T &  \dots &   A{\bf s}_t^T  \\
			| &  &|  \\		
		\end{array}\right]=\left[\begin{array}{ccc}
			| &    &   | \\
			{\bf r}_1^T+e_1{\bf b}^T &  \dots &  {\bf r}_t^T+e_t{\bf b}^T  \\
			| &  &|  \\		
		\end{array}\right]\in {\mathbb{Z}}^{m\times t},$$  
   and that  for every $i$, $\textbf{s}_i \in {\mathcal{S}}'$ if $e_i=0$. If $e_i=1$, our assumption (\ref{assumption}) allows Alice to choose from the beginning the sets ${\mathcal{S}}$ and ${\mathcal{S}}'$, such that $\textbf{s}_i \in {\mathcal{S}}$ with large probability $(\approx 1)$.
\end{itemize}

\section{The corresponding digital signature}
The Fiat-Shamir (FS) transform \cite{FiatShamir}  creates a non interactive proof of system or in the case of three moves id schemes, a  digital signature. In the latter case, it combines a hash function  to create a digital signature scheme, which under some conditions of the id scheme, the derived digital signature is secure in the Random Oracle Model (ROM). Such a paradigm is provided in the construction of the Schnorr digital signature in \cite{Sc1} and also to the newer digital signatures$:$ Picnic \cite{Picnic} and Dilithium \cite{Dilithium}.

We also assume that ${\mathcal{S}}={\mathcal{S}}_{\alpha_1,\alpha_2,\dots,\alpha_k}(n,R)$ and ${\mathcal{S}}'={\mathcal{S}}_{\beta_1,\beta_2,\dots,\beta_k}(n,R),$ where $$\beta_i\leq \frac{1}{R\ln{2}} \ln(2^{\alpha_i R-1} +1)\ (i=1,2,...,k),$$
for some $k$ and $R$ positive integers. 
In order to apply FS transform, we generate a transcript of our id scheme, i.e. a triple (for simplicity consider $t=1$),
\begin{center}
({\texttt{Commit}}: ${\bf r}$,  {\texttt{Challenge}}: $e$, {\texttt{Response}}: ${\bf s}$).
\end{center}
Let ${\mathcal{C}}$ be the commitment space and ${\mathcal{M}}$ the message space. The {\texttt{Challenge}} is computed as $e={\mathcal{H}}({\bf r}||\text{msg})$ with ${\mathcal{H}}: {\mathcal{C}} \times {\mathcal{M}}\rightarrow \{0,1\}$ be a hash function modeled as a random oracle. A hash ${\mathcal{H}},$ with the previous property, generates the challenge randomly as the verifier does in the id scheme. In this way we remove the verifier and transform the id scheme to a non-interactive one-move scheme.
The signature of the derived scheme is $\sigma = ({\bf r}, {\bf s})$ and it is valid if the transcript $({\bf r},e={\mathcal{H}}({\bf r}||\text{msg}),{\bf s})$ pass the verification algorithm.
I.e. the digital signature protocol is as follows:\\\\
{\texttt{Generation algorithm $G$}}. (We use the notation of the id scheme). We generate a matrix $A$ $(m\times n)$, then we pick a vector ${\bf x}\in {\mathcal S}\subset {\mathbb{Z}}^n$ and we compute ${\bf b}^T=A{\bf X}^T.$ $G$ outputs the public key $pk=(A,{\bf b},{\mathcal{S}},{\mathcal{S}}',t,{\mathcal{H}}),$ where ${\mathcal{H}}:\{0,1\}^*\rightarrow \{0,1\}^t$ is our hash, $t$ is a positive integer, and the secret key $sk={\bf x}.$
\ \\\\
{\texttt{Sign algorithm}}. Let $msg$ the message we want to sign. The signing algorithm runs as follows:\\
$sign(sk,msg) :$ \\
({\bf a}). Generate  a $n\times t$ random matrix $\mathcal{K}$  with columns: ${\bf k}_1,...,{\bf k}_t$ in  ${\mathcal{S}}'$\\
({\bf b}). $$  {\mathcal{R}} \leftarrow A  \mathcal{K} =\left[\begin{array}{ccc}
			| &    &   | \\
			A{\bf k}_1^T &  \dots &  A{\bf k}_t^T  \\
			| &  &|  \	
		\end{array}\right]\in {\mathbb{Z}}^{m\times t},$$
({\bf c}). ${\bf e}\leftarrow {\mathcal{H}}({\mathcal{R}},msg)$ and \\
({\bf d}).  $ \textbf{s}_i \leftarrow \textbf{k}_i + e_i \textbf{x}\in {\mathbb{Z}}^n$ and consider the matrix, 
$${\mathbb{S}}\leftarrow\left[\begin{array}{ccc}
			| &    &   | \\
			{\bf s}_1^T &  \dots &  {\bf s}_t^T  \\
			| &  &|  
		\end{array}\right]\in {\mathbb{Z}}^{n\times t}.$$    
\ \\
{\it output :} $\sigma = ({\mathcal R} , {\mathbb{S}} ; msg)\in {\mathbb{Z}}^{m\times t} \times {\mathbb{Z}}^{n\times t} \times {\mathcal{M}},$ where ${\mathcal{M}}$ is the message space. The pair $({\mathcal{R}},{\mathbb{S}})$  is the signature of the message $msg.$ 
%The number of elements of the signature is $2tn.$ If $\ell(S)=\alpha, \ell(S')=\alpha'$ are the bits of the largest element in ${\mathcal{S}},{\mathcal{S}'}$ resp., we get that the binary length of the signature is upper bounded by $tn(\alpha+\alpha').$ 
\ \\\\
{\texttt{Verification algorithm}}. The signature verification algorithm accepts as input, the signature $({\mathcal{R}},{\mathbb{S}})$ and the message $msg.$ Then, it computes 
$${\bf e}\leftarrow {\mathcal{H}}({\mathcal{R}},msg)$$ and outputs {\texttt{True}} if $A{\bf s}_i^T={\bf r}_i^T+e_i{\bf b}^T$ and ${\bf s}_i\in {\mathcal{S}}'$ if $e_i=0,$ else ${\bf s}_i\in{\mathcal{S}}$, for all $i=1,2,...,t.$ Else, it outputs {\texttt{False}}. 
\\\\
 % For instance, if the id scheme is Honest-Verifier Zero-Knowledge (HVZK), then the resulting signature is secure in the ROM, see \cite{Abdalla,fiat2}. However, we  follow a different route. 
 In  \cite[Theorem 19.15]{BonehShoup} was proved the following : \\ 
 $\bullet$ If $G()$ is one way (which we assume it is, under compact knapsack assumption),\\
 $\bullet$ The id scheme is S-HVZK, and \\
 $\bullet$ It is sound\footnote{in \cite[Theorem 19.15]{BonehShoup} this notion is called  {\it secure under direct attacks.}}, \\ 
 then we get that the id scheme is secure under eavesdropping attacks. Now, from \cite[Theorem 19.16]{BonehShoup} : \\
 {\it{An id scheme, secure under eavesdropping attacks and having unpredictable commitments, it provides via FS-transform secure signature schemes.}} \\
 Thus, since our scheme is secure under eavesdropping attacks,  it is enough to prove that it has unpredictable commitments. Then, our signature will be secure under active attacks, i.e. EUF-CMA\footnote{Existential UnForgeability under Chosen Message Attack}.	 %This was proved in \cite[]{Abdalla}.  %So, it is enough to show that the id scheme is special HVZK.
\subsubsection{Unpredictable Commitments} This notion has to do with collisions on commitment space. We say that an id scheme has $\delta-$unpredictable commitments if the probability to find a collision in the commitments is $\le \delta.$  

We can upper bound this probability, using simple geometric and combinatorial arguments. In fact, we shall use the same arguments as in \cite[Section 6.1]{draz2}. The probability to find a collision in the commitment space, say ${\mathbb{P}}$, is upper bounded by 
\begin{equation}\label{rel.unpredicatble}
\frac{|\{{\bf X}\in {\mathcal{S'}}: A{\bf X}^T={\bf r}^T\}|}{|{\mathcal{S}}'|}.
\end{equation}
I.e. the probability ${\mathbb{P}}$ to find a collision in commitments space (that is, to find two nonces ${\bf k_1}, {\bf k}_2$ from ${\mathcal{S}}', 	$ such that $A{\bf k}_1^T=A{\bf k}_2^T$) is upper bounded by the probability to find a solution of $A{\bf X}^T={\bf r}^T$ (for fixed ${\bf r}=(r_1,...,r_m)$), if we choose ${\bf X}$ randomly from ${\mathcal{S}}'.$ But the last quantity (\ref{rel.unpredicatble}), is bounded above by
$$\frac{|I_{\beta_1R}^{n/k}|\cdots |I_{\beta_kR}^{n/k-1}|}{|I_{\beta_1R}^{n/k}|\cdots |I_{\beta_kR}^{n/k}|}=\frac{1}{|I_{\beta_kR}|}=\frac{1}{2^{\beta_k R}}.$$
To see this, we remark that there is a hyperplane that meets the {\it{box}} $S'$ to $n-$points. One such hyperplane is a face of ${\mathcal{S}'}.$ So the numerator of the previous probability is bounded above by $|I_{\beta_1 R}^{n/k}|\cdots |I_{\beta_kR}^{n/k-1}|,$ where $\beta_k=\min\{\beta_i\}.$ 

Choosing $R$ large enough we get that ${\mathbb{P}}$ is negligible, so the id scheme has unpredictable commitments.

\subsubsection{The parameters}
If we choose $t=80$ and $A\xleftarrow{\$} M_n(I_{R/8}),$ using the results of table \ref{Table:3}, we set $m=4$ and we consider the set  ${\mathcal{S}}={\mathcal{S}}_{1/4,1/2} (n=48,R=192).$ That is 
$$\alpha_1=1/4,\alpha_2=1/2.$$ We have to choose 
$$\beta_i\leq \frac{1}{R\ln{2}}\ln\big(2^{\alpha_iR-1}+1\big),\ i=1,2.$$ 
Having the previous constraints we pick,
$$\beta_1= \frac{1}{8}, \beta_2 = \frac{1}{4}.$$
Then $\beta_1R=24, \beta_2 R = 48.$ The probabilities, see (\ref{probability}), with the previous choice are  $<10^{-8}.$
After straightforward calculations, we conclude that the signature is $\sim 33$ Kilobyte\footnote{\url{https://github.com/drazioti/compact-knapsack/blob/main/parameters.py}}, the length of the public key is $\sim 230$ Bytes and the secret key $\sim 16$ Bytes. Of course someone choosing smarter the parameters may reduce the length of the signature. This may be a matter of future work. As far as the signing time, this is fast since it consists only from dot products of integer vectors.
%\subsection{Post Quantum Signatures}

%Dilithium (Level 5) (i.e. 256 bit security) is going to be standardized. It uses $~2.5KB$ public key and almost double the length of pk for the secret key,  and signature $~4.5KB.$ It is based in Fiat-Shamir transform with aborts. Its security based on module-LWE/SIS. 

\section{Conclusions}
 In this paper, we have presented an extension of a cryptographic id-scheme based on Compact Knapsack problem. Initially, we experimentally studied the compact knapsack problem in its general case, by providing lattice attacks. Furthermore,  we examined the problem of selecting the correct parameters. What is more, we constructed a three-move id scheme based on compact knapsack. We proved that this scheme has all the nice properties i.e., complete, special sound, special HVZK, and finally sound. A parallel version of it, allow us to reduce the communication complexity and so build a digital signature based on the previous id scheme by using the FS transform. Some study has been done on how to select the parameters, this is in fact the selection of suitable solution sets ${\mathcal{S}}$ and ${\mathcal{S}}',$ so that, the compact knapsack be difficult in these sets.
 
 Under some suitable choice of the parameters the id scheme is secure under active attacks. Then, the corresponding digital signature using FS transform is proved to be secure under ROM. Future research, could extend our study by searching the optimum scheme's parameters and focus on construction of an application based on our id-scheme.
 
We believe that the scheme is secure in the post quantum setting, Since the introduction of the compact knapsack problem, and up until today, there has not been an efficient quantum algorithm developed for solving this problem.

%{\bf Acknowledgment}.\\

\appendix

\section{Smith Normal Form}\label{appendixA}

\begin{definition}
A rectangular integer matrix $S$ with dimension $n$ is in Smith Normal Form (SNF) if and only if $S$ is diagonal, say ${\rm{diag}}(s_1,s_2,...,s_n)$, such that $s_i$ are positive integers (for all $i$) and $s_{i+1}|s_{i}$ for all $i=1,2,...,n-1.$ 
\end{definition}

\begin{theorem}(Smith, 1861 \cite{Smith}).
Let $A\in {\mathcal{M}}_{m\times n}({\mathbb{Z}})$ of rank $r.$ Then there is a diagonal integer matrix $D={\rm diag}(\lambda_1,\lambda_2,...,\lambda_r,0,...,0)$ ($m\times n$) with 
$$ \lambda_1|\lambda_2|\cdots |\lambda_{r}$$
and unimodular matrices $U\in GL_m({\mathbb{Z}})$ and $V\in GL_n({\mathbb{Z}})$ such that 
$$A=UDV.$$
The non zero diagonal elements $\lambda_1,...,\lambda_r$ of $D$ are called {\it{elementary divisors}} of $A$ and are defined up to sign.
\end{theorem}

Algorithms for computing HNF and SNF appeared first time in 1971, see \cite{bradley} (not with polynomial time complexity) and with polynomial time complexity in 1979 by Kannan and Bachem \cite{kannan},  and in 1976 by von zur Gathen and Sieveking \cite{von} and Frumken \cite{frumkin}. Also see \cite[Chapter 5]{good_book}

In other words, for $A\in {\mathcal{M}}_{m\times n}({\mathbb{Z}})$ there are unimodular matrices $P,Q$ such that 
\[
PAQ=\left[\begin{array}{cc}
   {\rm{diag}}(\lambda_1,...,\lambda_r)    &  {\bf 0}_{r\times (n-r)}  \\
   {\bf 0}_{(m-r)\times r} & {\bf 0}_{(m-r)\times (n-r)}
		\end{array}\right],
   \]
		where $r$ is the rank of $A$ and $\lambda_i\in{\mathbb{Z}}_{>0},\ \lambda_i|\lambda_{i+1}.$
		Storjohann \cite[Theorem 12]{storj} provided a deterministic algortithm for computing SNF with complexity,
		$$O(nmr^2\log^{2}(r||A||)+r^4\log^{3}(||A||))\ \text{bit\ operations},$$ where $||A||=\max\{|a_{ij}|\}.$
		
		Let ${\rm{SNF}}(A)=PAQ.$ 
		Smith Normal Form is very useful if we want to solve a linear system over the integers. I.e. we want to solve a linear system $AX={\bf b}^T$ in integers. We start from the system $AX={\bf b}^{T},$ where $A$ is $m\times n$ and ${\bf b}^T$ is $m\times 1$ matrices. Let ${\bf x}^T$ (a $n\times 1$ matrix) an integer solution. We rewrite the system as $D{\bf y}^T={\bf c}^T,$ where $D={\rm{SNF}}(A),$ $Q{\bf y}^T={\bf x}^T$ and ${\bf c}^T=P{\bf b}^T.$ It is easy to solve the diagonal system $D{\bf y}^T={\bf c}^T$ and to see if it has an integer solution (see \cite{Laze}, for the history of this problem). We proved the following.
\begin{proposition}
Let $A$ a $m\times n$ and ${\bf b}^T$ a $m \times 1$ integer matrices. Let $AX={\bf b}^T$ be a linear system over ${\mathbb{Z}}.$ We set $D={\rm{SNF}}(A)={\rm{diag}}(\lambda_1,...,\lambda_r)=PAQ$ and ${\bf c}^T=P{\bf b}^T,$ $r$ the rank of $A$.\\
$({\bf i})$. The system has a solution in integers if and only if $c_{r+1}=c_{r+2}=\cdots=c_{m}=0$ and $\lambda_i|c_i$ for $i=1,2,...,r.$\\
$({\bf ii})$. If the system has an integer solution, the general solution of $DY={\bf c}^T$ is described as
$$(y_1,...,y_n)=(c_1/\lambda_1,c_2/\lambda_2,....,c_r/\lambda_r,t_{r+1},...,t_{n}),\ t_i\in {\mathbb{Z}}\ (\text{free}).$$
So, the general solution of the initial system is 
$$(x_1,...,x_n)^T=Q(y_1,...,y_n)^T.$$
\end{proposition}
\begin{corollary}
The system $AX={\bf b}^T$ has integer solutions if and only if the elementary divisors of $A,$ say $\lambda_1,...,\lambda_r,$ divide the constants $c_1,...,c_r$ (resp.) where ${\bf c}^T=P{\bf b}^T.$
\end{corollary}		
			If ${\rm{SNF}}(A)=PAQ,$ then from \cite{newman}, it is proved that the last $n-r$ columns of $Q$ is a basis for the integer lattice $AX={\bf 0}.$ 
\end{document}